\documentclass[a4paper, 12pt]{article}

\usepackage{latexsym,amsthm,amsmath}
\usepackage[T1]{fontenc}
\usepackage{amssymb}
\usepackage{graphicx,graphpap}
\usepackage{color}
\usepackage{dsfont}
\newcommand{\vc}[1]{\boldsymbol{#1}}
\newcommand{\p}{\textrm{P}}

\newcommand{\e}{\textrm{E}}
\newcommand{\dem}{\noindent\textbf{Proof. \,}}
\newcommand{\cqfd}{\hfill $\square$}

\numberwithin{equation}{section}
\theoremstyle{plain}

\newtheorem{lem}{Lemma}[section]
\newtheorem{prop}{Proposition}[section]
\newtheorem{cor}{Corollary}[section]
\newtheorem{rem}{Remark}[section]

\begin{document}

\title{Fitting Markovian binary trees using global and individual demographic data}
\author{Sophie Hautphenne\footnote{The University of Melbourne and Ecole polythechnique f\'ed\'erale de Lausanne, sophiemh@unimelb.edu.au}, Melanie Massaro\footnote{Charles Sturt University, mmassaro@csu.edu.au}, and Katharine Turner \footnote{Ecole polythechnique f\'ed\'erale de Lausanne, katharine.turner@epfl.ch}}
\maketitle 

%
%
%
%

\begin{abstract}
We consider a class of branching processes called Markovian binary trees, in which the individuals lifetime and reproduction epochs are modeled using a transient Markovian arrival process (TMAP). We estimate the parameters of the TMAP based on population data containing information on age-specific fertility and mortality rates. Depending on the degree of detail of the available data, a weighted non-linear regression method or a maximum likelihood method is applied. We discuss the optimal choice of the number of phases in the TMAP, and we provide confidence intervals for the model outputs. The results are illustrated  using real data on human and bird populations.
\end{abstract}

\section{Introduction} Simple birth-and-death processes do not offer enough flexibility to model real biological populations in which the age of individuals impacts on their fertility and mortality rates. The memoryless property inherent to these models implies that individuals do not age. However, they are tractable and amenable to efficient parameter estimation. 
In this paper, we model the lifetime and reproduction epochs of individuals in a population using a \textit{transient Markovian arrival process} (TMAP). Roughly speaking, a TMAP is a point process in which the event rate depends on the state of an underlying transient Markov chain with $n$ transient states (also called \textit{phases}), and one absorbing phase. Each event in the TMAP corresponds to the birth of a child, and the absorption in phase 0 corresponds to the individual's death. The resulting continuous-time branching process, called \textit{Markovian binary tree} (MBT), is the matrix generalisation of the birth-and-death process. It allows for much more flexibility than the latter, while keeping an excellent computational tractability. 

Performance measures of MBTs include the extinction probability of the population, the distributions of the time until extinction, the population size at any given time, and the total progeny size until any given time. 
The MBT model has already been used to efficiently compare demographic properties of female families in different countries, see \cite{hautphenne2012markovian}. 
The motivation behind the present paper is to develop
the statistical tools necessary to fit an MBT to populations of
species for which detailed information about age-specific survival and reproductive rates of individuals is available. The model can then be used to calculate age-dependent demographic properties. By knowing the exact age of individuals of a population, its future survival probability can be assessed, which aids conservation management of endangered species.

We fit a TMAP to different types of datasets which may be available from demographic databases or from studying an animal population in the field. 
These datasets can have different degrees of detail. We distinguish between:
\begin{itemize}\item \textit{Global population data}, consisting of the \textit{average} age-specific fertility and mortality \textit{rates} over an entire population. This sort of data is usually provided in databases on human fertility and mortality.
\item  \textit{Individual demographic data}, consisting of data on age-specific fertility and mortality \textit{counts} for each individual in a population. This sort of data often exists for closely monitored animal species. Here we will use data from a highly threatened bird species, the Chatham Island black robin (\textit{Petroica traversi}) \cite{butlerblack,massaro_sain2013}.
\end{itemize}The parameter estimation method depends on the type of data which are available: in the global population case, we use a weighted non-linear regression method to fit the parameters, and in the individual demographic data case, we use a maximum likelihood method. We consider different validation methods to determine the optimal number of phases $n$ in the TMAP.
Once a value of $n$ is determined and an estimator is found for the model parameters, we derive confidence intervals for the model outputs.

We apply our results to two real-world examples. The weighted non-linear regression method is applied in human demography leading to an improvement of the Markovian model considered in \cite{hautphenne2012markovian}. The maximum likelihood method is applied to the black robin population providing important insights about the species demography to be further discussed from a conservation biology point of view in a subsequent paper.

The paper is organised as follows: in Section~\ref{TMMPP}, we introduce TMAPs and describe the special case that we shall focus on. In Section~\ref{average} we perform model parameter estimation based on the average age-specific fertility and mortality rates, and in Section~\ref{ind}, we estimate the parameters based on individual age-specific fertility and mortality counts. In Section~\ref{num} we apply each method on a real-world example.

\section{Transient Markovian arrival processes}\label{TMMPP}

Transient Markovian arrival processes (TMAPs) are two-dimensional Markovian processes $\{(N(t),
\varphi (t)): t \in \mathbb{R}^{+}\}$ on the state space
$\mathbb{N}\times \{0,1, \ldots, n\}$, where $n$ is finite, combining the \emph{level} process $\{N(t)\}$, which counts the number of arrivals in $[0,t]$, with the \emph{phase} process,  \{$\varphi(t)\}$, which is a continuous-time Markov
chain. 
The states $(k,0)$ are absorbing for all $k\geq 0$; the other states are transient. 

A TMAP is characterized by two $n \times n$ rate matrices ${ D}_0$
and ${ D}_1$ and a non-negative
$n\times 1$ rate vector $\vc{d}$. Feasible transitions are from $(k,i)$ to $(k,j)$, for $k \geq 0$ and $1 \leq i \not= j \leq n$ at the rate $({ D}_0)_{ij}$, or from  $(k,i)$ to $(k+1,j)$ for $1\leq i,j\leq n$ at the rate $({ D}_1)_{ij}$, or from $(k,i)$ to $(k,0)$ at rate $d_i$.  The first transitions (at rate $({ D}_0)_{ij}$) are {\em hidden}: the phase of the individual changes but the level is not incremented.  The second transitions (at rate $({ D}_1)_{ij}$) are {\em observable}: a birth (arrival) is recorded, and the state of the individual may or may not change. The third transitions (at rate $d_i$) indicate the termination of the individual's life.

The matrix ${ D}_1$ and the vector $\vc d$ are nonnegative, ${ D}_0$ has
nonnegative off-diagonal elements and strictly negative elements on
the diagonal such that $ { D}_0\, \vc{1}+{ D}_1\, \vc{1 } + \vc{d} = \vc{0
}$, where $\vc 1$ is an $n\times 1$ vector of ones.  One
also defines the initial probability vector $\vc\alpha = (\alpha_i)_{1\leq i\leq n}$, and we assume that $\vc\alpha \vc 1 =1$, so that
$\varphi(0)\not=0$ a.s. More details on TMAPs can be found in \cite{latouche2003transient}.

There is a total of $p=2n^2+n-1$ entries in the matrices $\vc\alpha, D_0, D_1,\vc d$ if no assumption is made on their structure. A special case of TMAP, called the \textit{acyclic transient Markov modulated Poisson process} (ATMMPP), assumes 
 \begin{itemize}
\item individuals start their lifetime in phase 1 with probability one,
\item they can only move from phase $i$ to phase $i+1$ or to phase 0, with respective rates $\gamma_i$ for $1\leq i \leq n-1$ and $d_i$ for $1\leq i \leq n$, 
\item while in phase $i$, they reproduce at rate $\lambda_i$ and do not make any simultaneous phase transition at reproduction time.
\end{itemize}  
With these assumptions we have
$\vc\alpha=[1,0,\ldots,0]$, $D_1=$diag$(\lambda_1,\ldots,\lambda_n)$, and the only non-zero entries of $D_0$ are $(D_0)_{i,i+1}=\gamma_i$ and $$(D_0)_{ii}=\left\{\begin{array}{ll}\lambda_i-d_i-\gamma_i,& 1\leq i\leq n-1,\\ -\lambda_i-d_i,& i=n.
\end{array}\right.$$ There is a total of $p=3n-1$ parameters in an ATMMPP.

The lifetime distribution of a TMAP is phase-type PH$(\vc\alpha,D_0+D_1),$ see \cite{latouche2003transient}. Due to the structure of $D_0$ and $D_1$ in the ATMMPP case, this corresponds 
to a Coxian distribution. Such distributions are important as any acyclic phase-type distribution has an equivalent Coxian representation. Therefore, in terms of the lifetime distribution, the ATMMPP does not impose much restriction compared to the general TMAP.

\section{Global population data}\label{average}\label{av}\label{average}

\subsection{Available data and model equivalent}
For this section, we assume the available data are
 estimates of the expected age-specific fertility rates, $\hat{b}_x$, and
 estimates of the expected age-specific mortality rates, $\hat{d}_x$,
where $x\in\{0, 1, 2,\ldots, M\}$ denotes the age, that is, the period of time $[x,x+1)$ during the lifetime, and $M$ is the maximal age for which data are available. The method developed in this section can be generalised to $\ell$-year age classes, details are provided in Section \ref{ext} of the Supplementary Material. 

The rates $\hat{b}_x$ and $\hat{d}_x$ are interpreted respectively as the expected number of offspring per year from a parent at age $x$ and the probability that an individual who reached age $x$ dies within the year.
We denote by $\bar{d}(x)$ and $\bar{b}(x)$ the equivalent quantities computed from the TMAP model.
These functions have the following analytic expression, the proof of which is provided in Section \ref{proof_quant} of the Supplementary Material.

\begin{prop}\label{quant} The age-specific mortality and fertility rates in a TMAP with phase transition rate matrix $D:=D_0 +D_1$ are respectively given by
\begin{eqnarray*}
\bar{d}(x)&=& \dfrac{\vc\alpha e^{D x}(I-e^D)\vc 1}{\vc\alpha e^{D x}\vc 1}\\ \bar{b}(x)& =& \dfrac{\vc\alpha e^{Dx} (I-e^{D})(-D)^{-1}D_1\vc 1}{\vc\alpha e^{D x}\vc 1}.
\end{eqnarray*}
\end{prop}


\subsection{Parameter estimation}
In \cite{ll07} only death rates were used to fit phase-type lifetime distributions. We extend this approach, estimating the model parameters by minimizing the sum of weighted squared errors
\begin{equation}\label{F}
F=\sum_{x=0}^{M}\left[ (\hat{d}_x-\bar{d}(x))^2+ (\hat{b}_x-\bar{b}(x))^2\right] \hat{S}_x,  
\end{equation} 
where the weights $\hat{S}_x$ are the observed probabilities of survival until age $x$, $$\hat{S}_x=(1-\hat{d}_0)(1-\hat{d}_1)\cdots (1-\hat{d}_{x-1}).$$ As age increases there may be less available data leading to higher variance. These weights balance the potential resulting heteroscedasticity. If the estimated age-specific rates $\hat{d}_x$ and $\hat{b}_x$ are computed as averages of $n_x$ uncorrelated raw observations, another simple choice of weights would be $W_x = n_x$. 

Since the functions $\bar{d}(x)$ and $\bar{b}(x)$ are non-linear in both the input variable $x$ and in the parameters of the TMAP, we are dealing with a weighted non-linear regression. If there is missing information in the data and no estimate exists for $\hat{d}_x$ or $\hat{b}_x$ for some age $x$, then we set the corresponding term $(\hat{d}_x-\bar{d}(x))^2$ or $(\hat{b}_x-\bar{b}(x))^2$ to zero in the sum.


\begin{rem} \emph{The function $\bar{d}(x)$ corresponds to the hazard rate at age $x$ in survival analysis. Several hazard models have been considered to fit mortality data, such as the Gompertz-Makeham, the Siler, and the Heligman-Pollard models \cite{Gage1993}. Similarly, several age-specific fertility models have been studied, including the Hadwiger, the Beta, and the Gamma models \cite{peristera2007modeling}. The functions $\bar{d}(x)$ and $\bar{b}(x)$ are not claimed to provide better fits than these models, however, as opposed to the known mortality and fertility models which are generally studied separately, $\bar{d}(x)$ and $\bar{b}(x)$ are performance measures coming from \emph{the same Markovian model}, and are thus optimised together. The estimated Markovian model then corresponds to the best model fitting both the mortality and fertility data, and can be used to make a complete demographic study of the population, as shown in \cite{hautphenne2012markovian}.}
\end{rem}

\subsection{Goodness of fit and optimal value of $n$}

In the global population data case we estimate the model parameters by minimizing the objective function \eqref{F}. Therefore, a natural choice for the \emph{mean square error} function is
  $$\mathrm{MSE}=E\left[\sum_{x=0}^{M}\left[ (\bar{d}(x)-\hat{\bar{d}}(x))^2+ (\bar{b}(x)-\hat{\bar{b}}(x))^2\right] \bar{S}(x)\right].  $$ Here  $\bar{d}(x)$, $\bar{b}(x)$, and $\bar{S}(x)$ are respectively the age-specific mortality function, the age-specific fertility function, and the age-specific survival function corresponding to the \textit{true} model, and $\hat{\bar{d}}(x)$ and $\hat{\bar{b}}(x) $ are the equivalent functions corresponding to the \textit{estimated} model. 
If we know the true model then the MSE can be estimated through resampling. Alternatively this could be estimated when we are given a collection of datasets each containing global population information.
  
The value of the MSE indicates of the goodness of fit of the model. When the true model is unknown we estimate the optimal number of phases $n$ by minimizing the MSE. When the true model is known, comparing the MSE for different values of $n$ informs us on the sensitivity of the output with respect to the number of phases, as illustrated in Section \ref{details} in the Supplementary Material.

   

\section{Individual demographic data}\label{ind}

\subsection{Available data}

In this section, we assume the data are individual age-specific fertility and mortality counts in successive age-classes of length $\ell>0$\footnote{Successive
age-classes are assumed of equal length, but though
computationally convenient, this is not essential.}. They consist of $N$ vectors (one for each individual) of the type
\begin{equation}\label{vectv}\vc v=[ 6,  \,   8,  \,   -2  , \,  9 ,  \,  0    , \,3 ,  \,  3  , \, -1  ],\end{equation} of variable length, whose entries $v_i, i\geq 1$ are interpreted as follows:
\begin{itemize}\item $v_i=k\in\{0,1,2,\ldots\}$ if the individual had $k$ offspring while in the age-class $[(i-1)\ell,i\ell)$,
\item $v_i=-1$ if the individual died in the previous age-class $[(i-2)\ell,(i-1)\ell)$, possibly after producing some offspring, and\item $v_i=-2$ if the individual was alive at the beginning of the age-class $[(i-1)\ell,i\ell)$ but there is no (or incomplete) information on her progeny in that age-class.
\end{itemize}

\subsection{Parameter estimation}

Based on a sample of i.i.d. individual life vectors $\{\vc v^{(1)},\ldots,\vc v^{(N)}\}$, we maximize the log-likelihood function 
\begin{equation}\label{loglik}\mathcal{L}(\vc\theta; \vc v^{(1)},\ldots,\vc v^{(N)})=\sum_{j=1}^N \log p(\vc v^{(j)}|\vc\theta),\end{equation} where $\vc\theta=\{\vc\alpha, D_0, D_1,\vc d \}$, and $p(\vc v^{(j)}|\vc\theta)$ is 
the probability of observing the individual life vector $\vc v^{(j)}$, under the model parameter $\vc \theta$.  

Let $K=\max_{i,j}\{v_i^{(j)}: 1\leq i,1\leq j\leq N\}$ be the maximum number of offspring per age-class among the individuals in the sample. The probabilities $p(\vc v^{(j)}|\vc\theta)$ can be written as matrix products involving the matrices and vectors $P(k)=(P_{ij}(k))$, $\vc p(k)=(p_i(k))$, $P=(P_{ij})$, and $\vc p=(p_i)$ defined as
\begin{eqnarray}\label{eq1}P_{ij}(k)&:=&P[N(\ell)=k,\varphi(\ell)=j|N(0)=0,\varphi(0)=i],\\p_{i}(k)&:=&P[N(\ell)=k,\varphi(\ell)=0|N(0)=0,\varphi(0)=i],\\P_{ij}&:=&P[\varphi(\ell)=j|\varphi(0)=i]=\sum_{k\geq 0}P_{ij}(k),\\\label{eq4}p_{i}&:=&P[\varphi(\ell)=0|\varphi(0)=i]=\sum_{k\geq 0}p_{i}(k),\end{eqnarray}for $1\leq i,j\leq n$ and $1\leq k\leq K$. 
As an illustrative example, consider the four life vectors $$\vc v^{(1)}=[2,3,1,-1],\;\;\vc v^{(2)}=[2,-2,1,-1],\;\;\vc v^{(3)}=[2,3],\;\;\vc v^{(4)}=[2,-2].$$ By conditioning on the phases of the TMAP at the boundaries of the successive $\ell$-year intervals, the probability of observing these vectors is
$$\begin{array}{rclcrcl}p(\vc v^{(1)}|\vc\theta)&= &\vc\alpha P(2) P(3)\,\vc p(1),&\;& p(\vc v^{(2)}|\vc\theta)&= &\vc\alpha P(2) P\, \vc p(1),\\p(\vc v^{(3)}|\vc\theta)&= &\vc\alpha P(2)\,[P(3)\vc1+\vc p(3)],&\; &
p(\vc v^{(4)}|\vc\theta)&= &\vc\alpha P(2)\vc 1.\end{array}$$
Note that if $\vc v$ is a vector of size $M+1$ with all entries equal to $-2$, then $p(\vc v|\vc\theta)=\vc\alpha P^{M}\vc1$ is the probability that the individual survives at least the first $M$ age-classes.

The quantities defined in \eqref{eq1}--\eqref{eq4} can be computed explicitly, as shown in the next proposition, whose proof is provided in Section \ref{proof} of the Supplementary Material.
\begin{prop}\label{bigprop}For $1\leq k\leq K $, the matrix $P(k)$ and the vector $\vc p(k)$ are given by
\begin{eqnarray*}P(k)&=&(1/k!) (\vc e_k\otimes I) \exp(\mathcal{M}\ell)(\vc e_1^\top\otimes I),\\\vc p(k)&=&(1/k!) (\vc e_k\otimes I)[I-\exp(\mathcal{M}\ell)](-\mathcal{M})^{-1}\,(\vc e_1^\top\otimes I) \,\vc d,\end{eqnarray*}where $\vc e_k$ is the $k$th unit row vector of size $K$, and $$\mathcal{M}=\left[\begin{array}{ccccc} D_0 & &&&\\ D_1&D_0 & &&\\ &2 D_1&D_0 & &  \\ & &\ddots & &\\ &&&K D_1 &D_0\end{array}\right]. $$  The matrix $P$ and vector $\vc p$ are given by
$$P=\exp(D\ell),\qquad \vc p=[I-\exp(D\ell)](-D)^{-1}\vc d.$$
\end{prop}

Our MLE method generalises the results of Davison and Ramesh \cite{davison1996} who considered the parameter estimation of Markov modulated Poisson processes in the binary data case $v_i=\mathds{1}_{\{N(i\ell)-N((i-1)\ell)\geq 1\}}$. To our knowledge, other parameter estimation methods for such processes are based on the observation of the successive inter-event times rather than on the number of events within successive time intervals, see for instance \cite{ramesh} and \cite{ryden}. As confirmed in Figure \ref{f1a}, when the length of the time intervals decreases to zero, the estimates obtained with our method converge to those obtained with the usual method based on the observation of the successive inter-event times.

\subsection{Goodness of fit and optimal number of phases}

We consider three different criteria for choosing the optimal value of the number $n$ of phases. These criteria are compared on numerical examples in Section \ref{details} in the Supplementary Material.

\subsubsection*{{Akaike Information Criterion (AIC)}} We choose the value of $n$ which minimizes the AIC defined as
$$\mathrm{AIC}=2p-2 \mathcal{L}(\hat{\vc\theta}; \vc v^{(1)},\ldots,\vc v^{(N)}),$$ where the number of parameters $p=3n-1$ an ATMMPP model with $n$ phases. This criterion deals with the trade-off between goodness of fit of the model and its complexity. One advantage is that it does not rely on the knowledge of the true model: the value of AIC can be directly computed from the log-likelihood of the estimated model given the data. 


\subsubsection*{{Cross-validation (CV)}} 

We perform a {$K$-fold cross-validation} over the data sample of individual life vectors (with typical value $K=5$). 
The idea is to randomly divide the data into $K$ equal-sized parts. We leave out part $k$, fit the model to the other
$K -1$ parts (combined), and then evaluate the likelihood of the left-out $k$th part (test set) under the estimated parameters. 
We choose the model {maximizing} the mean test likelihood obtained by averaging the results for $k = 1,2,\ldots,K$.
Similar to the AIC, this method does not require us to know the true model. 


\subsubsection*{{Mean squared integrated loss (MSIL)}}
Let $\mathcal{V}$ be the set of all life vectors with entries in $\mathbb{N}\cup\{-1\}$.  
Any life vector with at least one entry equal to $-2$ is interpreted as a (disjoint) union of vectors in $\mathcal{V}$. For any fixed number of phases $n$, and a given sample of life vectors $\{\vc v_1, \vc v_2, \ldots \vc v_N\}$, the MLE method is used to estimate the parameters $\vc\theta_n = \{\vc\alpha, D_0, D_1, \vc d\}$ of the TMAP model. From the estimate $\hat{\vc\theta}_n$ we define a corresponding probability mass function $\hat{f}_n(\cdot)$ over $\mathcal{V}$ as $$\hat{f}_n(\vc v)=p(\vc v|\hat{\vc\theta}_n),\quad \vc v\in \mathcal{V}.$$  

The optimal number of phases is the value of $n$ minimizing the \textit{mean squared integrated loss}, defined as
\begin{eqnarray}\nonumber\mathrm{MSIL}&=&\mathrm{E}\left[\sum_{\vc v\in\mathcal{V}}(f(\vc v)-\hat{f}_n(\vc v))^2\right]\\\label{MSIL}&=&\sum_{\vc v\in\mathcal{V}}f(\vc v)^2 -2\e[\sum_{\vc v\in\mathcal{V}}f(\vc v)\hat{f}_n(\vc v)] + \e[\sum_{\vc v\in\mathcal{V}}\hat{f}_n(\vc v)^2].\end{eqnarray}
Since the first term is independent of $n$, the value of $n$ minimizing ${\textrm{MSIL}}^*(n):=\e[\sum_{\vc v\in\mathcal{V}}\hat{f}_n(\vc v)^2]-2\e[\sum_{\vc v\in\mathcal{V}}f(\vc v)\hat{f}_n(\vc v)] $ also minimizes the MSIL. The problem therefore reduces to estimating ${\textrm{MSIL}}^*(n)$ for each $n$.

If the true model is known, then $f(\cdot)$ is known, and the expectations in ${\textrm{MSIL}}^*(n)$ can be estimated through resampling.
If the true model is unknown, then a $K$-fold cross-validation method can be applied to estimate ${\textrm{MSIL}}^*(n)$. Let $A_k$ and $B_k$ be the $k$th training set and test set, respectively. Let $\hat{f}_n^k(\cdot)$ denote the probability mass function estimator using $n$ phases and training set $A_k$. We have \begin{equation}\label{CV_MSIL}\e\left[\sum_{\vc v\in\mathcal{V}}\hat{f}_n(\vc v)^2\right]\approx \frac{1}{K}\sum_{k=1}^K \sum_{\vc v\in\mathcal{V}}\hat{f}_n^k(\vc v)^2,\end{equation} and since the sets $B_k$ are all drawn from the true distribution $f(\cdot)$, we have $$\e\left[\sum_{\vc v\in\mathcal{V}}f(\vc v)\hat{f}_n(\vc v)\right]\approx\frac{1}{K}\sum_{k=1}^K \frac{1}{|B_k|}\sum_{\vc v\in B_k}\hat{f}_n^k(\vc v).$$ 

The set of life vectors $\mathcal{V}$ being infinite, the sums in \eqref{MSIL} and \eqref{CV_MSIL} need to be modified in practice. For a given pair of integers $K$ and $M$, we partition the set $\mathcal{V}$ to form a new \emph{finite} set $\tilde{\mathcal{V}}_{K,M}$ such that
$\sum_{\vc v\in\mathcal{V}}f(\vc v)=\sum_{\tilde{\vc v}\in\tilde{\mathcal{V}}_{K,M}}f(\tilde{\vc v})=1.$ The vectors $\tilde{\vc v}\in\tilde{\mathcal{V}}_{K,M}$ are of length $M$ and have their entries in the finite set $\{-1,0,1,2,\ldots,K,K+1\}$, so that
$$|\tilde{\mathcal{V}}_{K,M}|=(K+2)\dfrac{(K+2)^{(M+1)}-1}{(K+1)}.$$
They define equivalence classes in $\mathcal{V}$ as follows:
\begin{itemize}
\item if $-1\leq\tilde{v}_i\leq K$ for all $1\leq i \leq M$, then $$\tilde{\vc v}:=\big\{\vc v\in \mathcal{V}: v_i=\tilde{v}_i,\;\textrm{for all}\;  i \in\{1,\ldots, M\}\big\},$$ in which case $f(\tilde{\vc v})=p(\tilde{\vc v}|\vc\theta);$
\item if 
$\tilde{v}_{i_1}=\ldots=\tilde{v}_{i_\ell}= K+1$ for some $1\leq i_1,\ldots i_\ell\leq M$, $\ell\geq 1$, then \begin{eqnarray*}\tilde{\vc v}&:=&\big\{\vc v\in \mathcal{V}:v_i=\tilde{v}_i  \;\textrm{for all}\; i\in\{1,\ldots, M\}\setminus\{i_1,\ldots i_\ell\},\\&& \qquad\textrm{and}\; v_{i}\geq K+1 \;\textrm{for all}\; i\in\{i_1,\ldots i_\ell\} \big\},\end{eqnarray*}in which case $f(\tilde{\vc v})$ is computed as given in the next Lemma.
\end{itemize} 
\begin{lem} For any $\tilde{\vc v}\in\tilde{\mathcal{V}}_{K,M}$ such that
$\tilde{v}_{i_1},\ldots,\tilde{v}_{i_\ell}= K+1$ for some indices $1\leq i_1,\ldots i_\ell\leq M$, $\ell\geq 1$, we have
\begin{equation}\label{sum}f(\tilde{\vc v})=\sum_{k_1,\ldots,k_\ell\in\{-2,0,1,\ldots,K\}}(-1)^{\ell+\sum_{i=1}^\ell \mathds{1}\{k_i=-2\}} \;p({\vc v}^*(k_1,\ldots,k_\ell)|\vc\theta),\end{equation}where the vector ${\vc v}^*(k_1,\ldots,k_\ell)$ is such that, for $1\leq i\leq M$,
$${ v}_i^*(k_1,\ldots,k_\ell)=\left\{\begin{array}{ll} \tilde{v}_i & \textrm{if } i\notin \{i_1,\ldots i_\ell\}\\ k_j & \textrm{if } i=i_j \textrm{ for some $1\leq j\leq \ell$.}\end{array}\right.$$
\end{lem}
\dem We know $f(\tilde{\vc v})=\sum_{k_1,\ldots,k_\ell\geq K+1} p({\vc v}^*(k_1,\ldots,k_\ell)|\vc\theta)$ by the definition of  $\tilde{\vc v}$ and $\vc { v}^*(k_1,\ldots,k_\ell)$.
This sum contains $\ell$ embedded sums of the form $\sum_{k_j\geq K+1}$, which we rewrite as $\sum_{k_j\geq K+1}= \sum_{k_j\geq 0}-\sum_{0\leq k_j\leq K}$. Using
$$\sum_{k_j\geq 0} p({\vc v}^*(k_1,\ldots, k_j,\ldots,k_\ell)|\vc\theta)= p({\vc v}^*(k_1,\ldots, -2,\ldots,k_\ell)|\vc\theta),$$ and rearranging the terms then lead to \eqref{sum}. \cqfd

\medskip

Replacing $\mathcal{V}$ by $\tilde{\mathcal{V}}_{K,M}$ results in a different version of the MSIL criterion which selects the best model capturing differences in the number of children less than or equal to $K$ over the first $M$ age-classes. It is clear that the larger the age-class length $\ell$, the smaller $M$ and the larger $K$ should be chosen in order for $\mathcal{V}$ to be well approximated by $\tilde{\mathcal{V}}_{K,M}$.
When $\ell=1$, a possible choice of the partitioning parameters is taking $M$ as the ceiling of the expected lifetime plus one, and $K+1$ as the maximal expected number of children per age-class, that is, \begin{equation}\label{MK1}M=\lceil \sum_{x\geq 1} \hat{S}_x\rceil +1,\quad\textrm{and}\quad K+1=\lceil \max_{1\leq x\leq M} \hat{b}_x\rceil.\end{equation} Another choice leading to smaller equivalence classes in $\mathcal{V}$ is 
\begin{equation}\label{MK2}M= \min\{x\geq 0: \hat{S}_x<p\} +1,\quad\textrm{and}\quad K+1=\lceil \max_{1\leq x\leq M} (\hat{b}_x + \hat{\sigma}_x)\rceil,\end{equation}where $1-p$ is a covering probability, and $\hat{\sigma}_x$ is the standard error of the age-specific fertility rate at age $x$.
Formulae for $\ell-$year age-classes are analogous.



\subsection{Confidence intervals for the model outputs} For any performance measure of the model $g(x,\vc\theta)$, such as the mortality or fertility functions at age $x$, empirical and theoretical pointwise confidence intervals can be constructed. 

%
 If the true model is known, the pointwise mean and standard deviation of $g(x,\hat{\vc\theta})$ can be estimated through resampling. This provides a confidence interval for each value of $x$, and the width of the resulting confidence band gives us an indication of the stability of the estimated model, given the true model.
If the true model is unknown, bootstrapping from the data sample substitutes resampling from the true model.

Asymptotic theoretical confidence intervals are found using the delta method,
$$g(x,\hat{\vc\theta})\sim \mathcal{N}(g(x,{\vc\theta}),\nabla g(x,\vc\theta)\, J(\hat{\vc\theta})^{-1}\,\nabla  g(x,\vc\theta)^\top),\quad \textrm{as $N\rightarrow\infty$},$$where
$$J(\hat{\vc\theta})=\left.-\dfrac{\partial^2 \mathcal{L}({\vc\theta})}{\partial\vc\theta\partial\vc\theta^\top}\right|_{{\vc\theta}=\hat{\vc\theta}}$$ is the observed information matrix. A $95\%$ pointwise confidence interval for $g(x,{\vc\theta})$ is then given by
\begin{equation}\label{tci}g(x,\hat{\vc\theta})\pm1.96 \sqrt{\nabla g(x,\hat{\vc\theta})\, J(\hat{\vc\theta})^{-1}\,\nabla  g(x,\hat{\vc\theta})^\top}.\end{equation}


\section{Numerical applications}\label{num}

In this section, the two parameter estimation methods are applied to three types of illustrative examples. We first analyse artificial examples in which we simulate ATMMPPs, then we estimate their parameters based on the simulations, we make a goodness of fit analysis and compare the different criteria for choosing the optimal number of phases. We provide a summary of the results here and refer to Section \ref{details} in the Supplementary Material for details. Next, we use real global female population data in different countries to estimate the model parameters. Finally, we fit an MBT using real individual demographic data on the black robin population, and we give a biological interpretation to our results.

We used the Matlab function \texttt{fmincon} to minimize the sum of weighted squared errors \eqref{F} in the global population data case, or to maximize the log-likelihood function \eqref{loglik} in the individual demographic data case, under the constraint of positive parameters. The algorithm requires an initial value (seed) for the parameters. A reasonable guess for the model parameters is 
$$\gamma_i=\dfrac{n}{M+1}\qquad\mbox{for $1\leq i\leq n-1$,} $$ and
$$\lambda_i=\dfrac{\sum_{x=0}^{M} \hat{b}_x}{M+1},\quad \mu_i=\dfrac{\sum_{x=0}^{M}\hat{d}_x}{M+1},\qquad\mbox{for $1\leq i\leq n$.} $$ In order to minimize the risk to converge to local extrema, we started the algorithm with 25 different seeds obtained by adding random noise to the above values, and we chose the optimal solution among all.

\subsection{Artificial examples}

We considered three examples of ATMMPPs with $n=3$ or $n=4$ phases and we applied the two parameter estimation methods on data samples constructed by simulating trajectories of these models. 
As expected, the fits corresponding to the individual demographic data are much closer to the real model, and are associated to smaller confidence bands, than those corresponding to the global population data.

We observe that the MSE does not seem to be a satisfactory criterion to determine the optimal number of phases as the real value of $n$ never minimizes the MSE on these examples.
In all cases, the AIC provides the correct answer most of the time, while the CV and MSIL show similar trends and slightly under-estimate the true value of $n$. The parameters $K$ and $M$ in the MSIL were chosen according to \eqref{MK1} and the criterion turned out not to be sensitive to this choice as the optimal value of $n$ is the same for neighbouring values of $K$ and $M$. All the details and figures can be found in Section \ref{details} in the Supplementary Material.

\subsection{Female families in different countries} In \cite{hautphenne2012markovian} the authors used real global population data on female human mortality and fertility rates corresponding to five-year ages-classes from different countries to fit MBTs with $22$ phases; in these models, each phase corresponds exactly to one age-class. Here we use the weighted non-linear regression method described in Section \ref{ext} to estimate the parameters of new MBTs with $22$ phases, and we  compare the model age-specific mortality and fertility curves resulting from both approaches.
Besides facilitating the comparison, considering the same number of phases allows us to start the optimisation algorithm with the most realistic initial parameter values given by the model values in \cite{hautphenne2012markovian}. We show the results for a supercritical country (Congo), an almost-critical country (USA) and a subcritical country (Japan) in Figure \ref{demog1}. We see that the new fits have substantially improved: the MSE is divided by a factor of 5.21 for Congo, 1.89 for the USA and 1.38 for Japan. We observe that the fits of the mortality curves are less satisfactory for the older ages; note that removing the weights (which are decreasing with age) do not improve the fits.

\begin{figure}[h!]
\begin{center}
\raisebox{4cm}{(a)}\includegraphics[angle=0,
width=11cm, height=5cm]{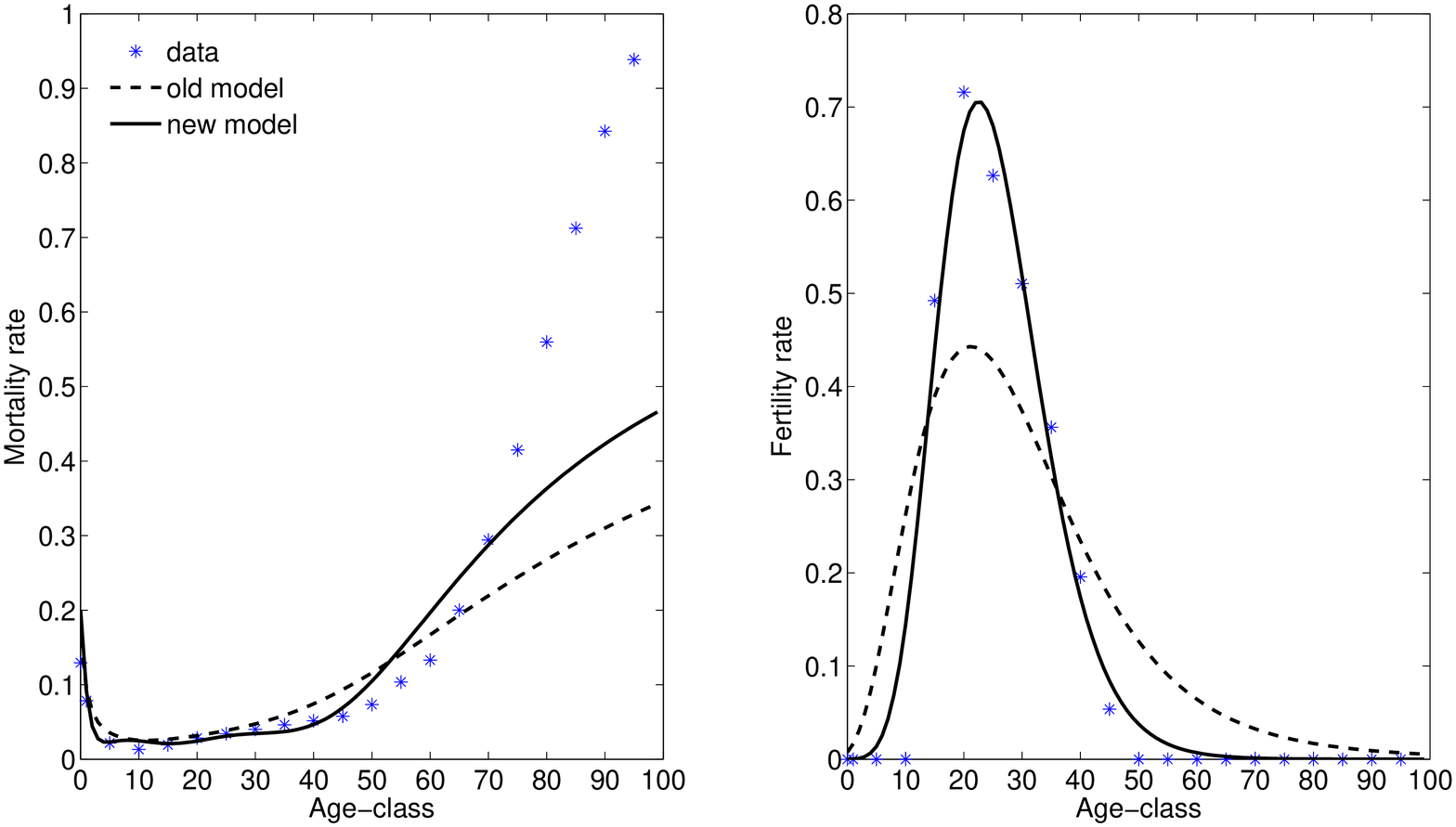}\vspace{-0.2cm}
\raisebox{4cm}{(b)}\includegraphics[angle=0,
width=11cm, height=5cm]{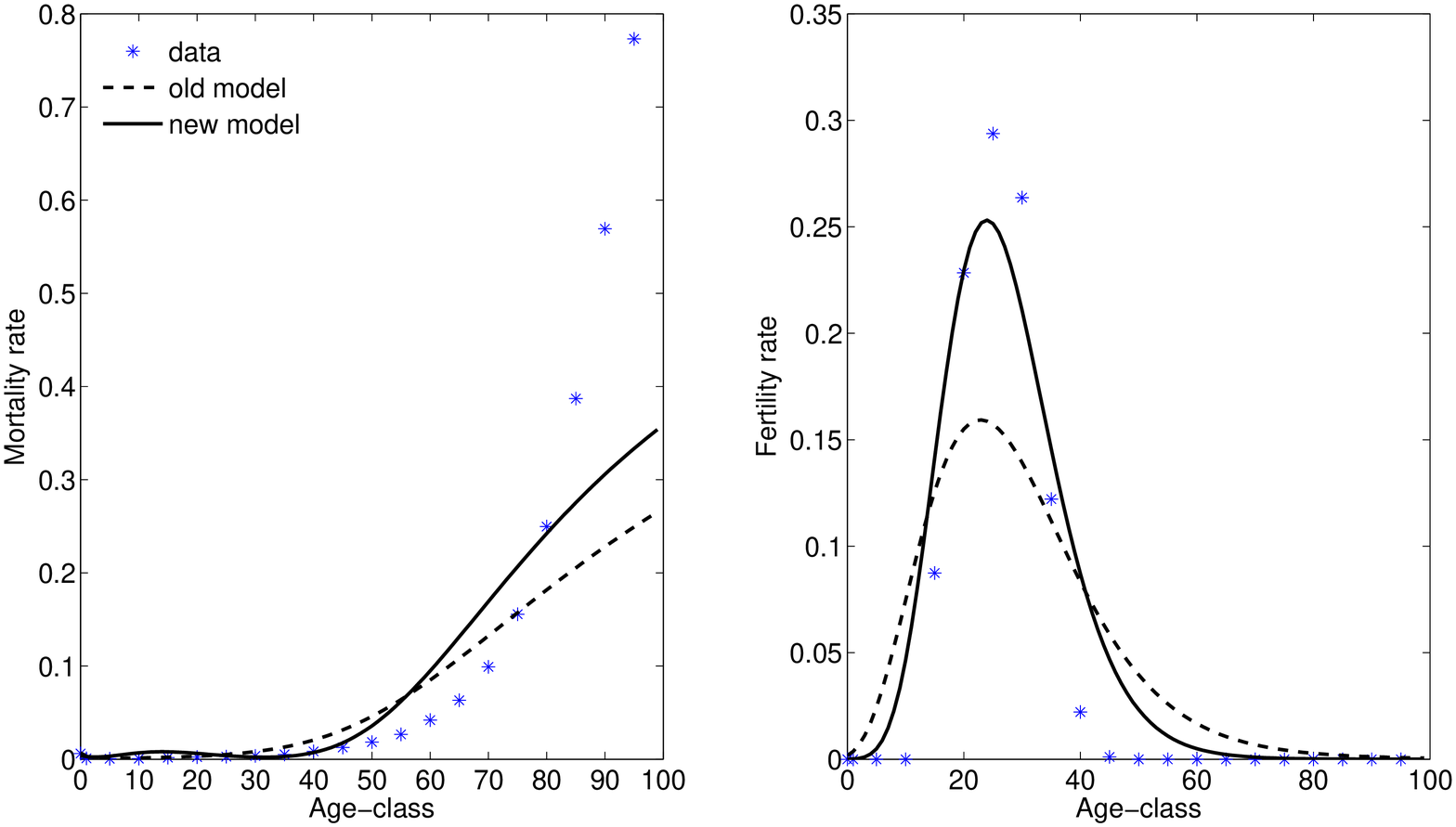}\vspace{-0.2cm}
\raisebox{4cm}{(c)}\includegraphics[angle=0,
width=11cm, height=5cm]{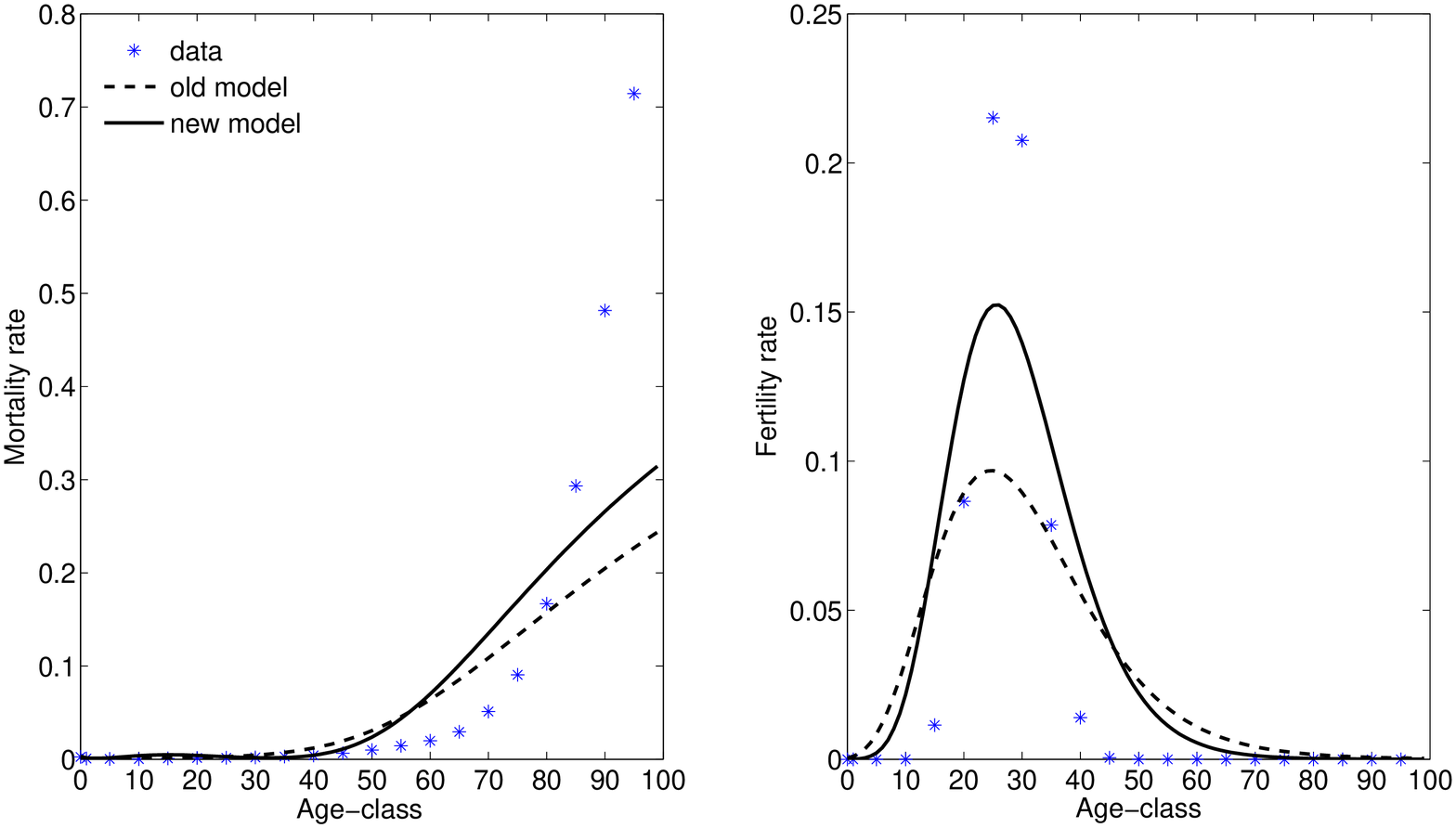}\vspace{-0.2cm}
\end{center}\caption{\label{demog1}\textbf{Female families.} Optimal model with $n=22$ computed using global population data (new model) corresponding to Congo (a), USA (b) and Japan (c), compared to the model used in \cite{hautphenne2012markovian} (old model).}
\end{figure}

\subsection{Black Robin population}

The black robin is an endangered songbird species endemic to the Chatham Islands, an isolated archipelago located 800 kilometres East of New Zealand. By 1980, the population of black robins had declined to five birds, including only a single successful breeding pair, on Mangere Island ~\cite{butlerblack}. Through intensive conservation efforts in 1980-1989 by the New Zealand Wildlife Service (now the Department of Conservation), the population recovered to 93 birds by spring 1990 \cite{kennedy2014severe}. Over the next decade (1990-1998), the population was closely monitored, but without human intervention. Nevertheless the population continued to grow rapidly to 197 adults by 1998, but after this period, the population growth slowed considerably and it only reached 239 adults in 2011 and 298 in 2014 \cite{massaro2013nest}.
 
For the conservation management of highly threatened species it is important to know the potential future viability, or survival probability, of a population, because if a population is not viable (i.e. fertility and survival rates are low), it will eventually become extinct. Population viability depends on reproductive rates and survival of individuals, but these rates may vary between sexes and across an individual's life span (with age). Hence, the exact male-to-female ratio and the ages of each individual within a population will influence a population's future viability. Reintroduction of a species into previously occupied parts of its former natural range is nowadays a common hands-on conservation method. In order to maximize the survival chances of the new population, it is necessary to know the optimal age distribution of the reintroduced population, which 
can only be designed based on a complete age-specific demographic analysis of the species. The black robin is an ideal species for which to develop these new statistical tools, because biologists have been collecting complete raw datasets on this bird species for several decades.
 An age-specific reorganisation of these data leads to a total of 
433 life vectors for the monitoring period 2007-2014.

\begin{figure}[h!]
\begin{center}
\includegraphics[angle=0,
width=11cm, height=5cm]{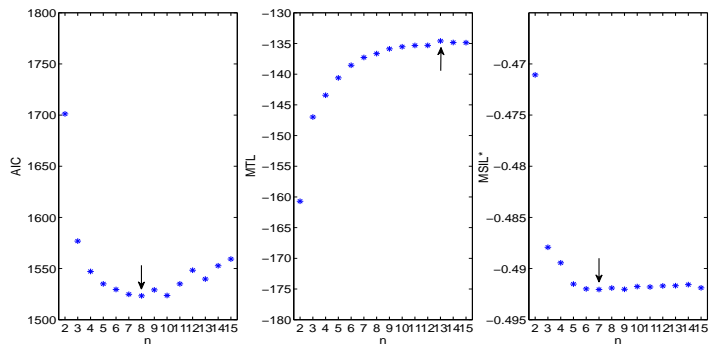}
\end{center}
\caption{\label{BR_AIC}\textbf{Black robins.} Left: Result of AIC: the optimal number of phases is $n=8$; Middle: Result of CV: the optimal number of phases is $n=13$; Right: Result of MSIL with $M=3$ and $K=1$: the optimal number of phases is $n=7$.}
\end{figure}

\begin{table}\begin{center}
\begin{tabular}{c|cccc} $M\smallsetminus K$ & 0 & 1 & 2 & 3\\\hline 2 & 5 & 9 & 13 & 13\\3 & 5 & \textbf{7} & 14 & 14\\ 4 & 5 & 15 & 14 & 14\end{tabular}\end{center}
\caption{\label{MSIL_BR}\textbf{Black robins.} The optimal number of phases according to the MSIL criterion for different values of $M$ and $K$.}
\end{table}

We performed different tests to determine the optimal number of phases to fit the black robin data. The results are shown in Figure \ref{BR_AIC}, where we see that the optimal value is $n=8$ according to the AIC, $n=13$ according to the CV criterion, and $n=7$ according to the MSIL criterion with $M=3$ and $K=1$ (determined using \eqref{MK1}). In this case, the MSIL criterion is quite sensitive to the choice of $M$ and $K$, as indicated in Table \ref{MSIL_BR}. 

\begin{figure}[h!]
\begin{center}
\includegraphics[angle=0,
width=13cm]{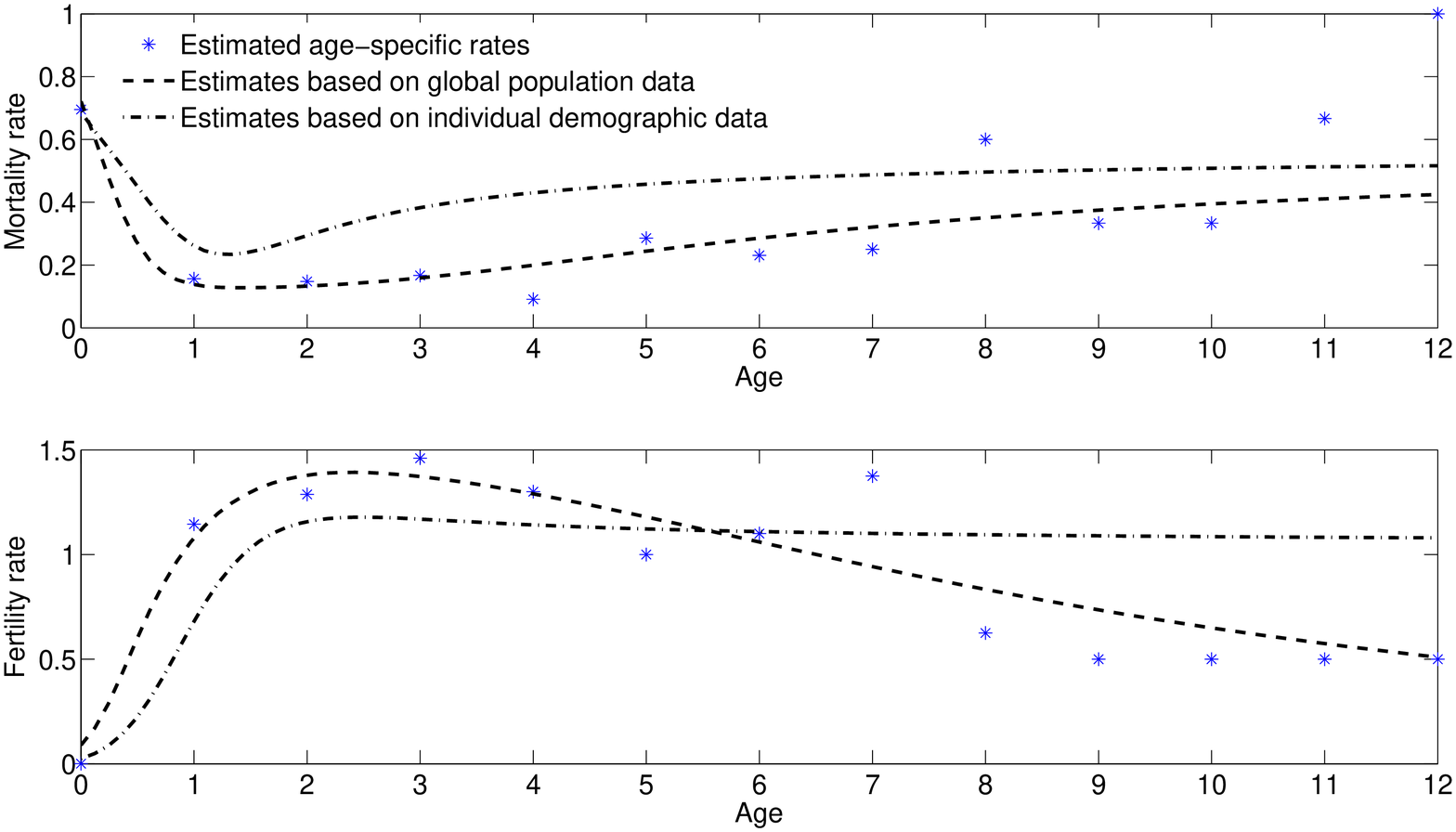}
\caption{\label{BR_p3}\textbf{Black robins.} Age-specific mortality and fertility curves for the models based on global population data and on individual demographic data with $n=8$.
}
\end{center}
\end{figure}

The model fits based on the global population data and on the individual demographic data (life vectors) with $n=8$ are compared in Figure \ref{BR_p3}. Since life vectors are available here, the corresponding models are the most informative.
Black robins reach maturity at 1 or 2-years of age. The age-specific mortality curves show that mortality of black robins is the highest  before they reach maturity and the lowest when they are 1 to 2 years old. Once birds reach 3 years of age, mortality rates do not increase dramatically with age, nor do fertility rates decline, which would support the hypothesis that there is no senescence.
However, as few birds reach the old ages, the accuracy of the estimates obtained using global population data declines with age. Figure \ref{BR_cip2} shows the 95$\%$ pointwise confidence intervals for the estimates of the model outputs, obtained by bootstrapping 25 samples from the original sample of life vectors.
 We see that the confidence intervals are quite narrow, especially for the estimates obtained using the individual demographic data.
 
 \begin{figure}[h!]
\includegraphics[angle=0,
width=13cm]{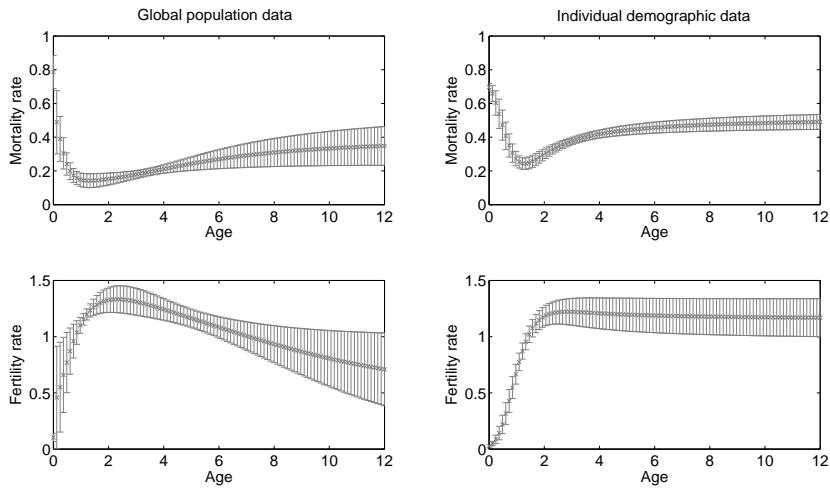}
\caption{\label{BR_cip2}\textbf{Black robins.} Mean and $95\%$ pointwise confidence intervals of the model fits corresponding to 25 {bootstrapped datasets} generated from the individual dataset containing $N=433$ life vectors, for $n=8$, using global population data (right) and individual demographic data (left).}
\end{figure}

\begin{figure}[h!]
\begin{center}
\includegraphics[angle=0,
width=11cm]{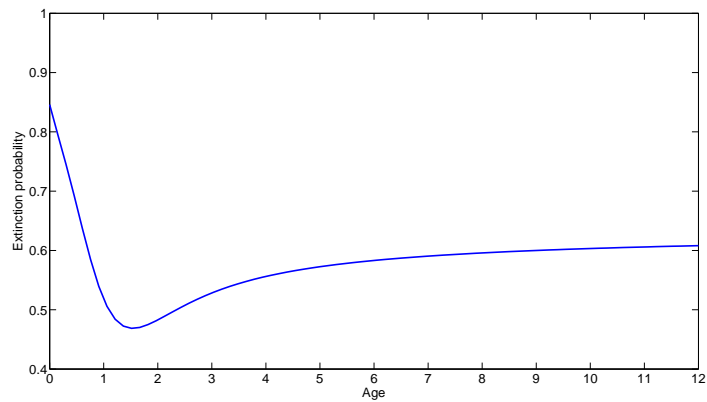}
\caption{\label{BR_comp2}\textbf{Black robins.}  The extinction probability of a female population as a function of the age of the initial female computed from the model with $n=8$ estimated using  individual demographic data.
}
\end{center}
\end{figure}

One of the most informative model output is the probability of extinction of a female family (that is, consisting only of female descendants) generated by a singe female, as a function of the age of that first female. This probability can be computed from the MBT model using any of the available algorithms (see for instance \cite{soph2,soph1}) and is shown in Figure \ref{BR_comp2}. The curve highlights the combined effect of the age-specific mortality and fertility rates on the viability of the female family, hence of the whole population, by extension. We see that, if we were to found a new population starting with a single female bird, in order to maximize the survival chance of the population, the optimal age of the initial female should be around one year old.

  
  \section{Future directions}
There are a number of directions for future research, particularly for the study of global population data. Cross-validation is delicate in that case as there are generally few data points. Leave-one-out cross-validation could be used, where we leave one age-class out. Possible complications would then include the choice of the weights $\hat{S}_x$ as these explicitly use the death rates for all age-classes.

Further methods for analysing global population data could be developed for when the age-specific mortality and fertility rates follow some specific distributions. For example, we may know that the birth rate and death rate at each age lie in an exponential family of distributions. The analysis could then involve a parallel process of estimating these distributions and using these distributions to simulate new samples of global population data.
Each new sample can then be used to do parameter estimation and construct confidence intervals.
 Alternatively, the theory developed for finding confidence intervals for weighted non-linear regression methods could be used in this case; unfortunately, this process is not computationally straightforward.

\section*{Acknowledgements}
\small{Sophie
Hautphenne thanks the Australian Research Council (ARC) for support through the
Discovery Early Career Researcher Award DE150101044.
The black robin research was funded by the New Zealand Foundation for Research, Science and Technology (UOCX0601) to Melanie Massaro, by the School of Biological Sciences, University of Canterbury, by the Brian Mason Scientific and Technical Trust, and by the Mohamed bin Zayed Species Conservation Fund. This research was only possible with permission from the Chatham Island Conservation Board and the logistic help of the Department of Conservation. 
}


\clearpage

\section{Supplementary Material}

\subsection{Proof of Proposition \ref{quant}} \label{proof_quant}
Let $L$ be the lifetime of an individual and $\bar{S}(x)=P[L>x]$ be the survival function. Since $L$ follows a PH$(\vc\alpha, D)$ distribution, $$\bar{S}(x)=1-P[L\leq x]=\vc\alpha e^{D x}\vc 1.$$ 
The probability of death at age $x$, $\bar{d}(x)$, can thus be calculated as 
\begin{eqnarray*} \bar{d}(x) &=&P[x< L \leq x+1 | L> x]= \dfrac{P[L>x]-P[L> x+1]}{P[L>x]}\\&=& \dfrac{\bar{S}(x)-\bar{S}(x+1)}{\bar{S}(x)}=\dfrac{\vc\alpha e^{D x}(I-e^D)\vc 1}{\vc\alpha e^{D x}\vc 1}.
\end{eqnarray*} It is shown in \cite{latouche2003transient} that the mean number of events until time $t$ in a TMAP started in phase $i$ at time 0 is given by $$E[N(t)|\varphi(0)=i]=[(I-e^{D t})(-D)^{-1}D_1\vc 1]_i.$$ 
Let $N([x,x+t))$ denote the number of events in the TMAP in the time interval $[x,x+t)$. By time-homogeneity of the TMAP, $$E[N([x,x+t))|\varphi(x)=i]=E[N(t)|\varphi(0)=i].$$
The mean number of offspring generated by an individual at age $x$ can thus be calculated as
\begin{eqnarray*} \bar{b}(x)& =& E[N([x,x+1)) | L>x]\\ &=&\sum_{1\leq i\leq n} P[\varphi(x)=i|L>x] E[N([x,x+1)) | L>x, \varphi(x)=i]\\&=&\sum_{1\leq i\leq n} \dfrac{P[\varphi(x)=i, L>x]}{P[L>x]}E[N(1)|\varphi(0)=i]\\&=& \sum_{1\leq i\leq n}\dfrac{[\vc\alpha e^{Dx}]_i}{\vc\alpha e^{D x}\vc 1}[(I-e^{D})(-D)^{-1}D_1\vc 1]_i\\
&=& \dfrac{\vc\alpha e^{Dx} (I-e^{D})(-D)^{-1}D_1\vc 1}{\vc\alpha e^{D x}\vc 1}.
\end{eqnarray*}\cqfd

\subsection{Global population data with $\ell-$year age classes}\label{ext}


In demography, the available data often consist of age-specific fertility and mortality rates over $\ell-$year age-classes with $\ell>1$, that is,
\begin{itemize}
\item the expected number of offspring \textit{per year} from a parent in age-class $[x,x+\ell)$, denoted as $\hat{\beta}_{[x,x+\ell)}$, and
\item the probability that an individual who reached the age-class $[x,x+\ell)$ dies \textit{within the year}, denoted as $\hat{\mu}_{[x,x+\ell)}$.
\end{itemize}


We extend the definition of $\bar{d}(x)$ and $\bar{b}(x)$ to $\ell-$year age-classes and define the functions $\bar{d}(x,\ell)$ and $\bar{b}(x,\ell)$, computed from the TMAP model, as follows:
\begin{eqnarray*}\bar{d}(x,\ell) &=&P[x< L \leq x+\ell | L> x]\\
\bar{b}(x,\ell)& =& E[N([x,x+\ell)) | L>x].\end{eqnarray*} It is a simple matter to generalise Proposition \ref{quant} to obtain
\begin{cor}\label{cor1} 
\begin{eqnarray*}
\bar{d}(x,\ell)&=& \dfrac{\vc\alpha e^{D x}(I-e^{D\ell})\vc 1}{\vc\alpha e^{D x}\vc 1}\\ \bar{b}(x,\ell)& =& \dfrac{\vc\alpha e^{Dx} (I-e^{D\ell})(-D)^{-1}D_1\vc 1}{\vc\alpha e^{D x}\vc 1}.
\end{eqnarray*}
\end{cor} 

The correspondence between $\bar{d}(x,\ell)$ and $\bar{b}(x,\ell)$ and the mortality and fertility {rates} in age-class $[x,x+\ell)$ is given in the next Lemma.
\begin{lem} \label{lem1} 
\begin{eqnarray*}\bar{d}(x,\ell) &\equiv &1-(1-\hat{\mu}_{[x,x+\ell)})^\ell\\
\bar{b}(x,\ell)& \equiv&\hat{\beta}_{[x,x+\ell)} \,\dfrac{1-(1-\hat{\mu}_{[x,x+\ell)})^\ell}{\hat{\mu}_{[x,x+\ell)}},\end{eqnarray*}where the symbol $\equiv$ has to be interpreted as ``is the model equivalent of''.
\end{lem}

\begin{proof}The model function $1-\bar{d}(x,\ell)$ is the probability that an individual who reached age $x$ survives at least until age $x+\ell$, that is, survives $\ell$ successive one-year age intervals, which occurs with probability $(1-\hat{\mu}_{[x,x+\ell)})^\ell$.

The model function $\bar{b}(x,\ell)$ can be rewritten as
\begin{eqnarray*}\bar{b}(x,\ell)&=&E[N([x,x+1)) | L>x]+E[N([x+1,x+2)) | L>x]+\ldots\\&&+E[N([x+\ell-1,x+\ell)) | L>x]\\&\equiv& \hat{\beta}_{[x,x+\ell)}  + (1- \hat{\mu}_{[x,x+\ell)}) \hat{\beta}_{[x,x+\ell)} +\ldots + (1-\hat{\mu}_{[x,x+\ell)} )^{\ell-1} \hat{\beta}_{[x,x+\ell)} \\&=&\hat{\beta}_{[x,x+\ell)} \,\dfrac{1-(1-\hat{\mu}_{[x,x+\ell)})^\ell}{\hat{\mu}_{[x,x+\ell)}},
\end{eqnarray*}which completes the proof.
\end{proof}

In order to estimate the model parameters in this case, the objective function \eqref{F} then needs to be modified according to Corollary \ref{cor1} and Lemma \ref{lem1}.

\subsection{Proof of Proposition \ref{bigprop}}\label{proof}

In order to compute $P(k)$, we actually compute $P(k,t)$ for $0\leq k\leq K$, and for any $t\geq 0$ where $$P_{ij}(k,t):=P[N(t)=k,\varphi(t)=j|N(0)=0,\varphi(0)=i],$$ and observe that $P(k)=P(k,\ell)$. It is well known from the theory of MAPs that the probability generating function $P^*(z,t):=\sum_{k\geq 0} P(k,t) z^k,$ is given by the matrix exponential
$$P^*(z,t)=\exp[D(z)t],\quad \textrm{where $D(z):=D_0+z\,D_1.$}$$Since $P(k,t)=(1/k!)[\partial^k P^*(z,t)/ (\partial z)^k]\big|_{z=0}$ for any $k\geq 0$, we need to take the derivatives of the matrix exponential $\exp[D(z)t]$ with respect to $z$. The scalar rule of exponential differentiation only holds here if $D_0$ and $D_1$ commute, which is generally not the case. Instead, we first differentiate $P^*(z,t)$ with respect to $t$,
\begin{equation} \partial P^*(z,t) /\partial t= D(z) P^*(z,t),\end{equation} and we then take successive derivatives of this equation with respect to $z$:
\begin{eqnarray*}
 \partial^2 P^*(z,t) /(\partial t)(\partial z)&=&D_1 P^*(z,t)+ D(z) \partial P^*(z,t) /\partial z  \\\partial^3 P^*(z,t) /(\partial t)(\partial z)^2&=&2 D_1 \partial P^*(z,t) /\partial z+ D(z) \partial^2 P^*(z,t) /(\partial z)^2 \\&\vdots &\\\partial^{(K+1)} P^*(z,t) /(\partial t)(\partial z)^K&=&K D_1 \partial^{(K-1)} P^*(z,t) /(\partial z)^{(K-1)}\\&&+D(z) \partial^K P^*(z,t) /(\partial z)^K .
\end{eqnarray*} This system of partial derivative equations can be rewritten as an ordinary differential equation for the unknown matrix containing the partial derivatives of $P^*(z,t)$ with respect to $z$,
$$\small{\dfrac{d}{dt}\left[\begin{array}{c}P^*(z,t)\\\partial P^*(z,t)/ \partial z\\\partial^2 P^*(z,t)/ (\partial z)^2\\ \vdots\\ \partial^K P^*(z,t)/ (\partial z)^K\end{array}\right] = \left[
\begingroup 
\setlength\arraycolsep{2pt}
\begin{array}{ccccc} D(z)& &&&\\ D_1&D(z)& &&\\ &2 D_1&D(z) & &  \\ & &\ddots & &\\ &&&K D_1 &D(z)\end{array}
\endgroup
\right] \left[\begin{array}{c}P^*(z,t)\\\partial P^*(z,t)/ \partial z)\\\partial^2 P^*(z,t)/ (\partial z)^2\\ \vdots\\ \partial^K P^*(z,t)/ (\partial z)^K\end{array}\right]},  $$ whose solution is
$$\small{\left[\begin{array}{c}P^*(z,t)\\\partial P^*(z,t)/ \partial z\\\partial^2 P^*(z,t)/ (\partial z)^2\\ \vdots\\ \partial^K P^*(z,t)/ (\partial z)^K\end{array}\right]=\exp\left(\left[
\begingroup 
\setlength\arraycolsep{2pt}
\begin{array}{ccccc} D(z) & &&&\\ D_1&D(z) & &&\\ &2 D_1&D(z) & &  \\ & &\ddots &\phantom{D(z)} &\\ &&&K D_1 &D(z)\end{array}
\endgroup
\right]\,t\right) \left[\begin{array}{c}I\\0\\0\\ \vdots\\ 0\end{array}\right]}.$$ Taking $z=0$ and denoting
$$\mathcal{M}=\left[\begin{array}{ccccc} D_0 & &&&\\ D_1&D_0 & &&\\ &2 D_1&D_0 & &  \\ & &\ddots & &\\ &&&K D_1 &D_0\end{array}\right], $$ we obtain $$P(k)=P(k,\ell)=(1/k!) (\vc e_k\otimes I) \exp(\mathcal{M}\ell)(\vc e_1^\top\otimes I).$$ Then, $\vc p(k)$ is obtained by conditioning on the time $u\in[0,\ell]$ when the individual dies,
\begin{eqnarray*}\vc p(k)&=&\int_0^\ell P(k,u) \,\vc d\, du\\&=&(1/k!) (\vc e_k\otimes I)\int_0^\ell \exp(\mathcal{M} u)\,du\,(\vc e_1^\top\otimes I) \,\vc d\\&=&(1/k!) (\vc e_k\otimes I)[I-\exp(\mathcal{M}\ell)](-\mathcal{M})^{-1}\,(\vc e_1^\top\otimes I) \,\vc d\end{eqnarray*} Next, $$P=P^*(1,\ell)=\exp(D\ell),$$ where $D=D_0+D_1$, and finally $$\vc p=\int_0^\ell  \exp(Du)\vc d\,du=[I-\exp(D\ell)](-D)^{-1}\vc d.$$ \cqfd

\subsection{Further details on the artificial examples}\label{details}

We consider three examples of ATMMPPs, simulating $N$ trajectories of these processes for $T$ units of time.
The different parameter values are summarized in Table \ref{ta1}.

%
\begin{table}[h!]\centering
\begin{tabular}{l|c|ccc|cccc|cccc|c|c}
&$n$&$\gamma_1$&$\gamma_2$&$\gamma_3$ & $\mu_1$ & $\mu_2$&$\mu_3$&$\mu_4$&$\lambda_1$&$\lambda_2$ & $\lambda_3$ &$\lambda_4$ &$N$ & $T$ \\\hline Ex. 1& 3&$0.25$&$0.25$&$-$& $0.2$&$0.4$& $0.9$&$-$& 6 & 3 & 2 &$-$&$500$ & $15$\\ Ex. 2& 4&$0.5$&$0.1$&$0.1$& $0.3$&$0.1$& $0.2$&$0.7$& $0.5$ & 2 & $0.5$ &$0.01$& $400$ & $25$\\Ex. 3&4&$0.3$&$0.3$&$0.3$& $0.6$&$0.1$& $0.2$&$0.5$& $0.2$ & 3 & 2 &$0.1$ & $500$ & $15$
\end{tabular} \caption{\label{ta1}Parameter values and simulation characteristics of the ATMMPPs corresponding to the three artificial examples. The number of phases is denoted by $n$, $\gamma_i$ is the transition rate from phase $i$ to phase $i+1$, $\mu_i$ is the death rate in phases $i$, $\lambda_i$ is the birth rate in phase $i$, $N$ is the number of simulated trajectories, and $T$ is the simulation time for each trajectory.}
\end{table}

For each example, we associated each simulated trajectory of the ATMMPP with a life vector by counting the number of births falling in successive $\ell$-year intervals. Here we took $\ell=1$, so each entry of the vectors corresponds to a specific age-classes and the vectors have a total of $T$ entries. This produced samples $\{\vc v^{(1)},\ldots,\vc v^{(N)}\}$ of $N$ individual life vectors of the form \eqref{vectv}. The average age-specific fertility and mortality rates $\hat{b}_x$ and $\hat{d}_x$, for $0\leq x\leq M=T-1$, were computed directly from these samples.


We first performed a goodness of fit analysis on Example 1. We set $n=3$ (the true number of phases), and we used $(i)$ the global population data $\hat{d}_x$ and $\hat{b}_x$, and $(ii)$ the full sample of $N$ life vectors, to estimate the model parameters using the corresponding statistical method, leading to two different parameter estimates $\hat{\vc\theta}^{(i)}$ and $\hat{\vc\theta}^{(ii)}$. In Figure \ref{f1bis} we compare the performance measures $g(x,\hat{\vc\theta}^{(i)})$ and $g(x,\hat{\vc\theta}^{(ii)})$ corresponding to the age-specific mortality and fertility curves, obtained with the two different estimation methods.
  To further assess the accuracy of the estimates, we re-sampled 50 datasets of size $N$ from the true model, and we show in Figure \ref{f3bis} the mean curves compared to the real ones, as well as the corresponding $95\%$ pointwise confidence intervals. In Figure \ref{f3bis2}, we perform the same analysis by bootstrapping 50 times from a single dataset instead of resampling. We conclude from Figures \ref{f1bis}, \ref{f3bis} and \ref{f3bis2}  that, as expected, the fits corresponding to the individual demographic data are much closer to the real model, and are associated to smaller confidence bands, than those corresponding to the global population data. In Figure \ref{ex1_CI} we show the theoretical $95\%$ pointwise confidence intervals given by \eqref{tci} for the fits obtained using individual demographic data; these are comparable to those shown in Figure \ref{f3bis}. Finally, in
Figure \ref{f1a} we compare the fits based on life vectors with different age-class lengths $\ell$ to those based on the observation of the successive inter-event times. We see that as $\ell$ decreases to zero, the estimates obtained with our method converge to those based on the successive inter-event times.


For all three examples, Figure \ref{MSE_all} shows the sensitivity of the MSE with respect to the number of phases in the fitted model. This also highlights the fact that the MSE does not seem to be a satisfactory criterion to determine the optimal number of phases as the real value of $n$ never minimizes the MSE on these examples.  Finally, Figure \ref{all_crit} compares all criteria to decide upon the optimal number of phases in the individual demographic data case. In all cases, the AIC provides the correct answer most of the time, while the CV and MSIL show similar trends and slightly under-estimate the true value of $n$. The parameters $K$ and $M$ in the MSIL were chosen according to \eqref{MK1} and the criterion turned out not to be sensitive to this choice as the optimal value of $n$ is the same for neighbouring values of $K$ and $M$.


\begin{figure}[h!]
\begin{center}
\includegraphics[angle=0,
width=13cm]{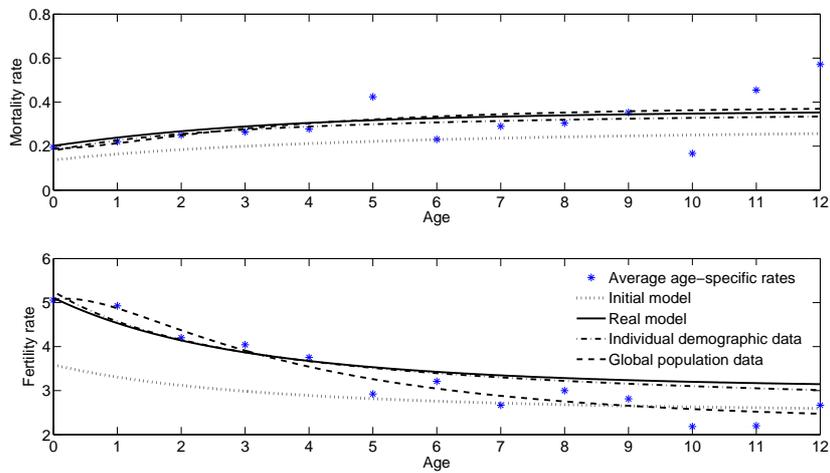}
\caption{\label{f1bis}\textbf{Example 1.} Comparison of the model fits obtained using global population data and individual demographic data. The initial model is the one used as a seed in the optimisation algorithms.}\end{center}
\end{figure}


\begin{figure}[h]
\begin{center}
\includegraphics[angle=0,
width=13cm]{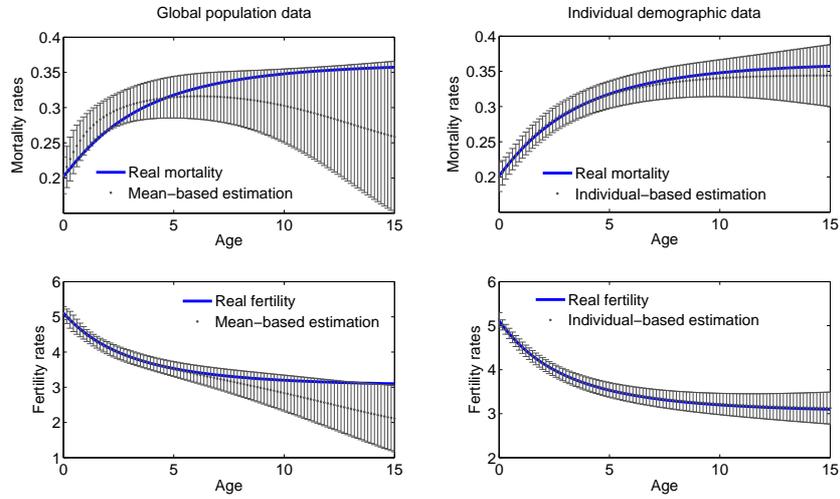}
\end{center}\caption{\label{f3bis}\textbf{Example 1.} Mean and $95\%$ pointwise confidence intervals of the model fits corresponding to 50 simulations from the real model using global population data (left) and individual demographic data (right).}
\end{figure}

\begin{figure}[h]
\begin{center}
\includegraphics[angle=0,
width=13cm]{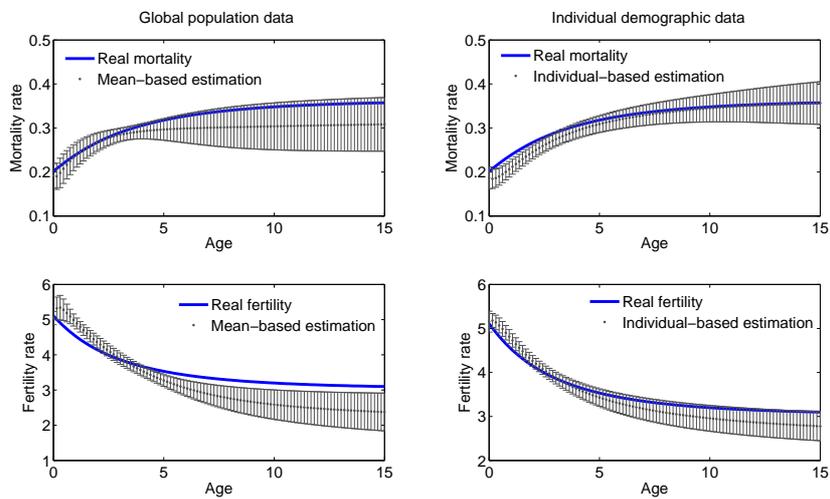}
\end{center}\caption{\label{f3bis2}\textbf{Example 1.} Mean and $95\%$ pointwise confidence intervals of the model fits corresponding to 50 {bootstrapped datasets} using global population data (left) and individual demographic data (right).}
\end{figure}

\begin{figure}[h!]
\begin{center}
\includegraphics[angle=0,
width=13cm]{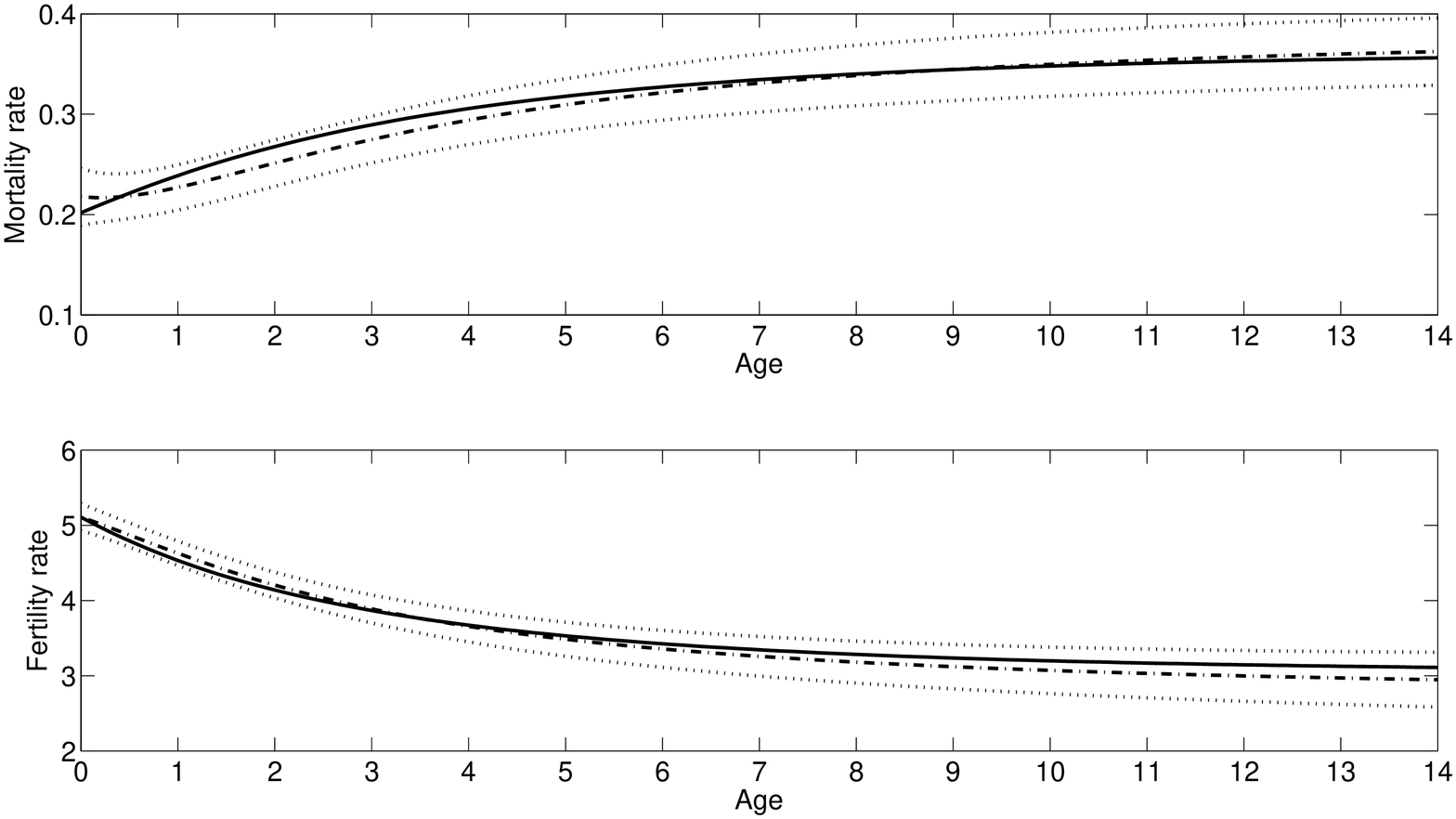}
\caption{\label{ex1_CI}\textbf{Example 1.}{ Theoretical} $95\%$ pointwise confidence intervals (dotted lines) in the individual demographic data case. The true curves correspond to the plain lines and the estimated curves correspond to the dash-dot lines.}
\end{center}
\end{figure}

\begin{figure}[h!]
\begin{center}
\includegraphics[angle=0,
width=13cm]{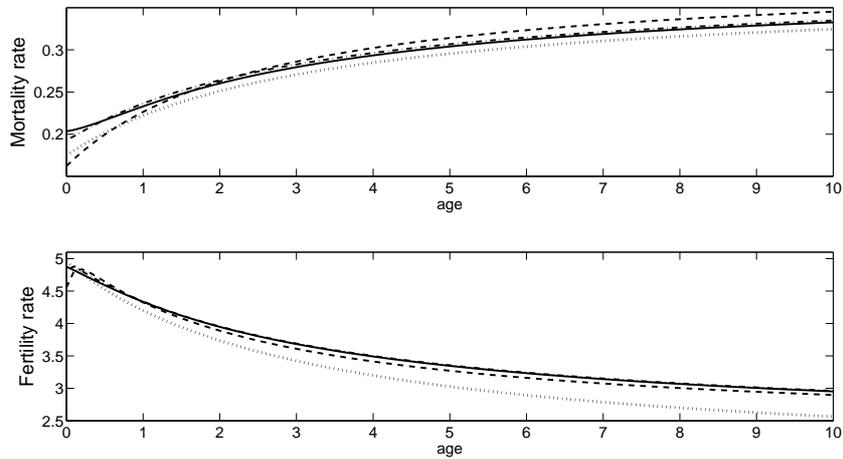}
\end{center}\caption{\label{f1a}\textbf{Example 1.} Comparison of the model fits obtained using the MLE based on inter-arrival times (plain lines), and individual demographic data with $\ell=5$ (dotted lines), $\ell=2.5$ (dashed lines), and $\ell=1$ (dash-dot lines).}
\end{figure}

\begin{figure}[h!]
\begin{center}
\includegraphics[angle=0,
width=12cm, height=5cm]{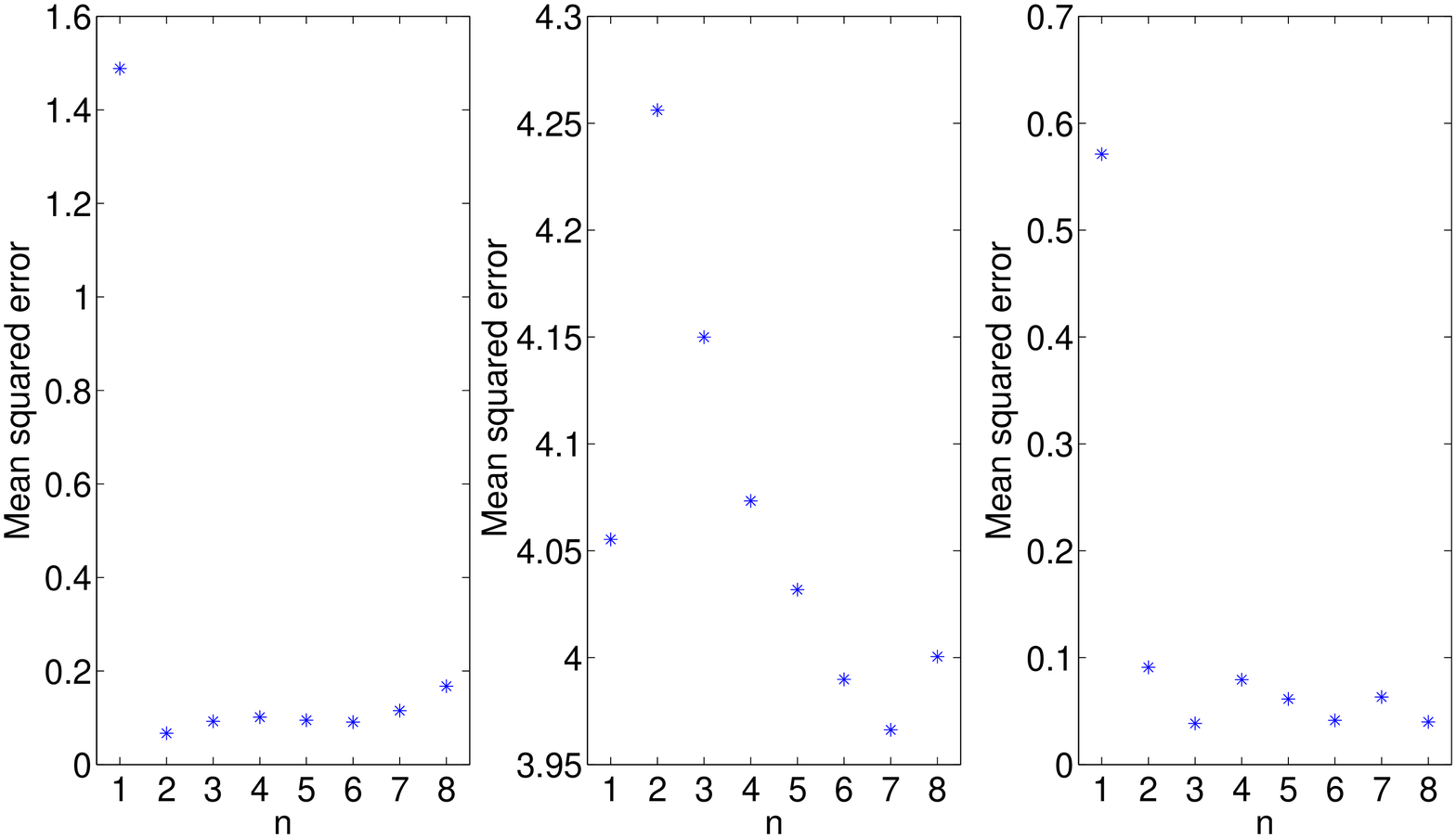}
\caption{\label{MSE_all}\textbf{Examples 1 (left), 2 (middle), and 3 (right).} Mean squared error based on 50 simulations from the true models.}
\end{center}
\end{figure}

\begin{figure}[h!]
\begin{center}
\raisebox{3.5cm}{(a)}\includegraphics[angle=0,
width=10cm, height=4.5cm]{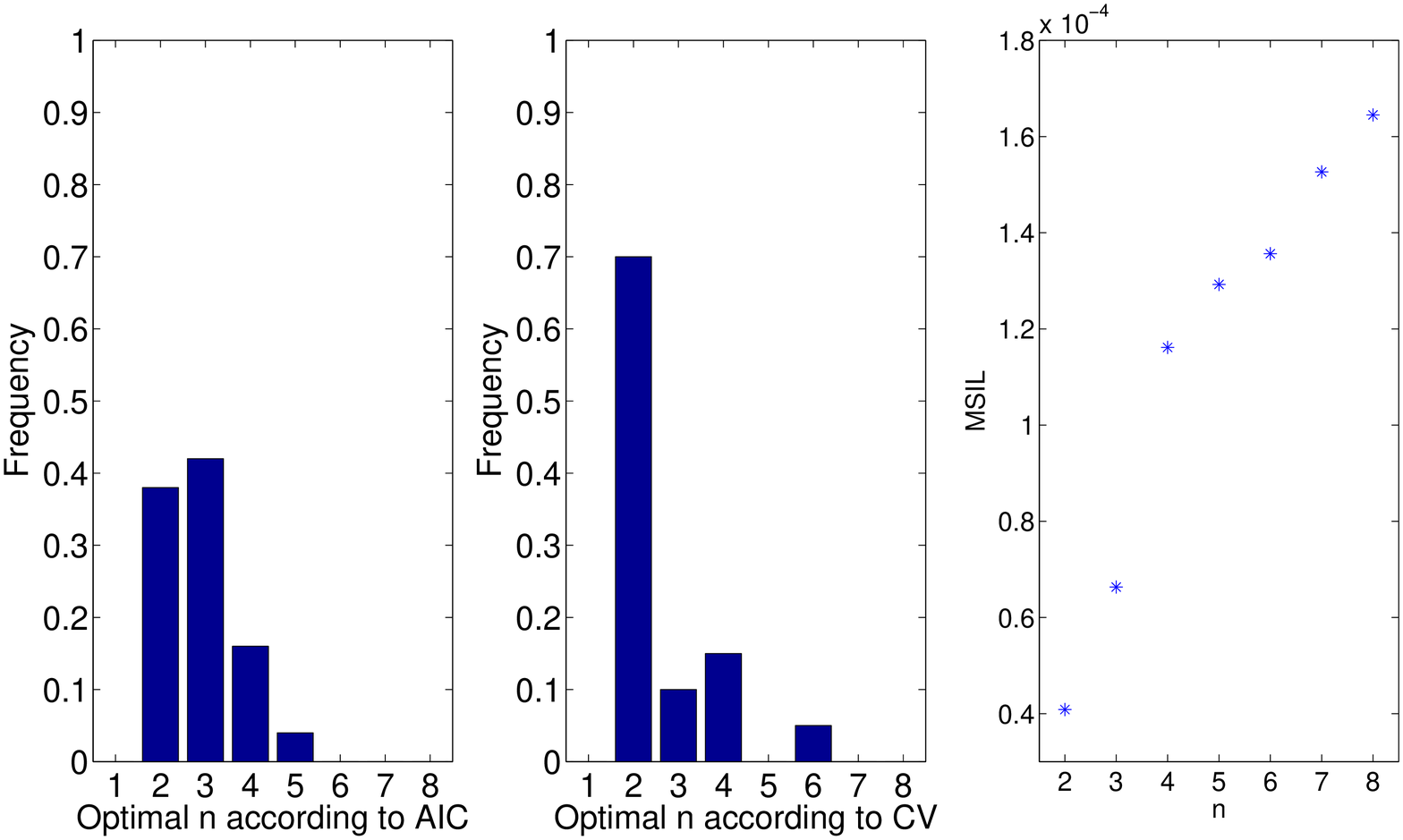}\vspace{-0.15cm}\\
\raisebox{3.5cm}{(b)}
\includegraphics[angle=0,
width=10cm, height=4.5cm]{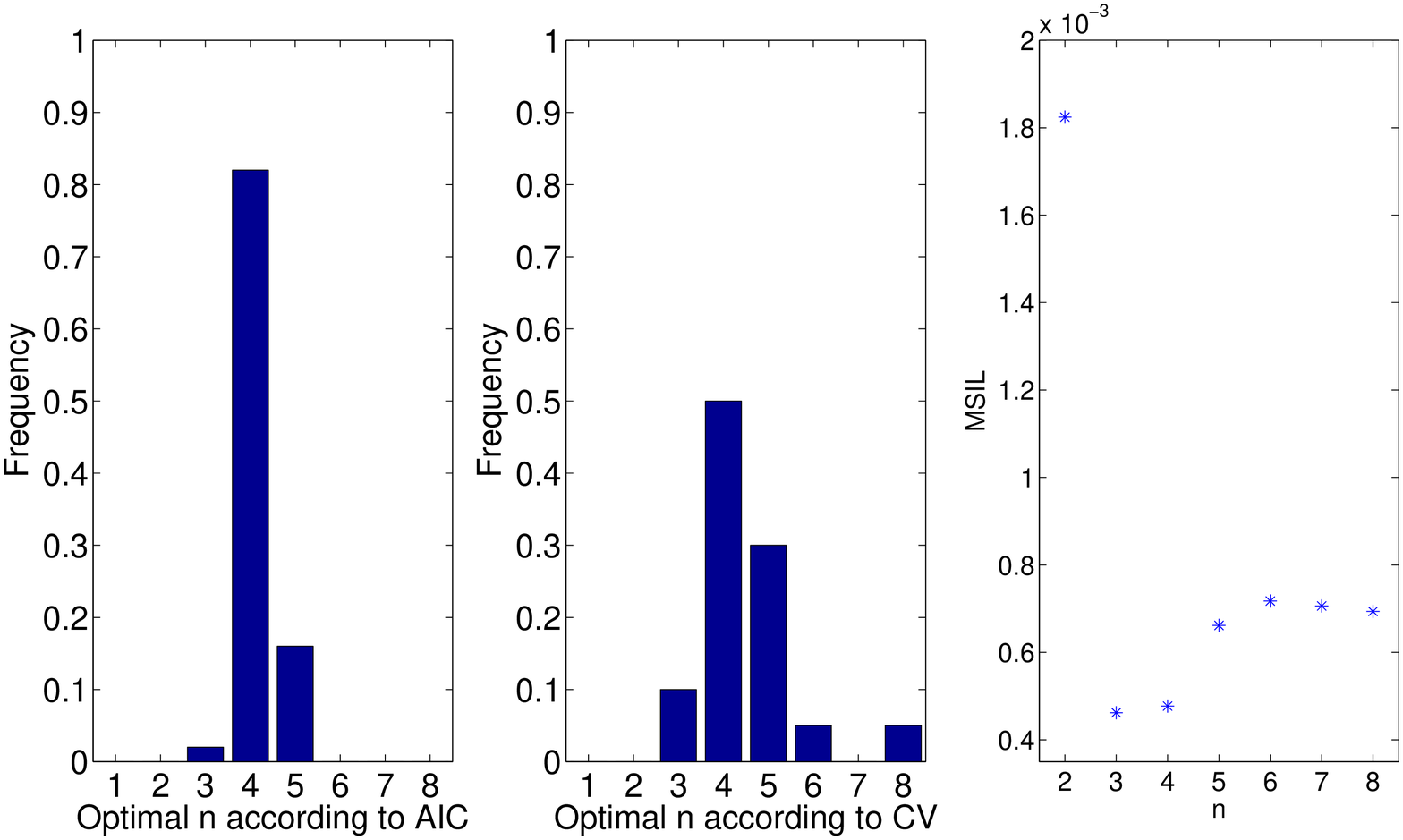}\vspace{-0.15cm}
\raisebox{3.5cm}{(c)}\includegraphics[angle=0,
width=10cm, height=4.5cm]{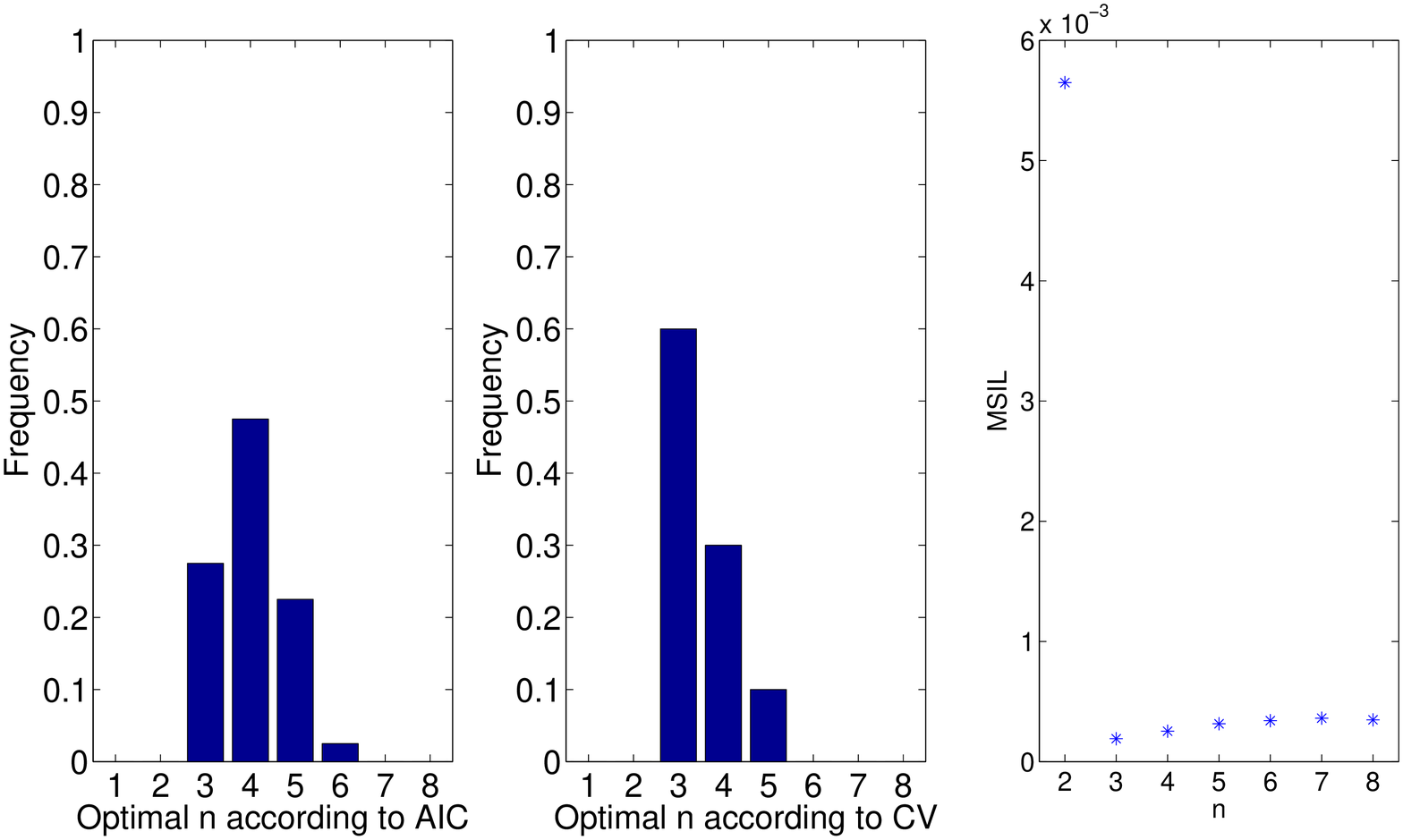}\vspace{-0.15cm}
\caption{\label{all_crit}\textbf{Examples 1 (a), 2 (b), and 3 (c).} Left: Frequency of optimal $n$ according to AIC based on 50 simulations from the true model. Middle:  Frequency of optimal $n$ according to CV based on 20 simulations. Right: MSIL for $2\leq n\leq 8$ based on 50 simulations from the true model, with $M=4$ and $K=5$ (a), $M=3$ and $K=1$ (b), and $M=3$ and $K=2$ (c). We omit the value at $n=1$ which is much larger than the value at $n=2$.}
\end{center}
\end{figure}


\begin{thebibliography}{9}

\bibitem{butlerblack}
D.~Butler and D.~Merton.
\newblock {\em The Black Robin: Saving the World's Most Endangered Bird}.
\newblock Oxford University Press, Auckland, 1992.

\bibitem{davison1996}
A. C. Davison and N. I. Ramesh.
\newblock Some models for discretized series of events.
\newblock {\em Journal of the American Statistical Association}, 91(434): 601--609, 1996.

\bibitem{Gage1993}
T. B. Gage and C. J. Mode. \newblock Some laws of mortality: how well do they fit?  \newblock {\em Human Biology}, 65(3) 445--461, 1993.

\bibitem{hautphenne2012markovian}
S. Hautphenne and G. Latouche.
\newblock The {M}arkovian binary tree applied to demography.
\newblock {\em Journal of mathematical biology}, 64(7):1109--1135, 2012.

\bibitem{soph2}
S.~Hautphenne, G.~Latouche, and M.-A. Remiche.
\newblock Newton's iteration for the extinction probability of a {M}arkovian
  {B}inary {T}ree.
\newblock {\em Linear Algebra and its Applications}, 428:2791--2804, 2008.

\bibitem{soph1}
S.~Hautphenne, G.~Latouche, and M.-A. Remiche.
\newblock Algorithmic approach to the extinction probability of branching
  processes.
\newblock {\em Methodology and Computing in Applied Probability}.
\newblock 13(1), 171-192, 2011.

\bibitem{kennedy2014severe}
E.~S. Kennedy, C.~E. Grueber, R.~P. Duncan, and I.~G. Jamieson.
\newblock Severe inbreeding depression and no evidence of purging in an
  extremely inbred wild species---the {Chatham Island} black robin.
\newblock {\em Evolution}, 68(4):987--995, 2014.

\bibitem{latouche2003transient}
G. Latouche, M.-A. Remiche, and P. Taylor.
\newblock Transient {M}arkov arrival processes.
\newblock {\em Annals of Applied Probability}, pages 628--640, 2003.

\bibitem{ll07}
X.S. Lin and X.~Liu.
\newblock Markov aging process and phase-type law of mortality.
\newblock {\em North American Actuarial Journal}, 11:92--109, 2007.


\bibitem{massaro_sain2013}
M. Massaro, R. Sainudiin, D. Merton, J.V. Briskie, A.M. Poole, and M.L. Hale.  \newblock Human-assisted spread of a maladaptive behavior in a critically endangered bird. \newblock {\em PloS one}, 8(12), e79066, 2013.

\bibitem{massaro2013nest}
M.~Massaro, M.~Stanbury, and J.~V.~Briskie.
\newblock Nest site selection by the endangered black robin increases
  vulnerability to predation by an invasive bird.
\newblock {\em Animal Conservation}, 16(4):404--411, 2013.

\bibitem{merton}
D~Merton.
\newblock The chatham island black robin.
\newblock {\em Forest and bird}, 21(3):14--19, 1990.

\bibitem{peristera2007modeling}
P. Peristera and A. Kostaki.
 \newblock Modeling fertility in modern populations.
  \newblock {\em Demographic Research},
  16(4):141--194, 2007.

\bibitem{ramesh}
N. I. Ramesh.
\newblock Statistical analysis on Markov-modulated poisson processes.
\newblock {\em Environmetrics}, 6(2):165--179, 1995.

\bibitem{ryden}
T. Ryd\'en.
\newblock Parameter estimation for Markov modulated Poisson processes.
\newblock {\em Stochastic Models}, 10(4):795--829, 1994.

\end{thebibliography}
\end{document}